\documentclass{article}
\usepackage[utf8]{inputenc}

\usepackage{subfigure}
\usepackage{tikz}
\usepackage{amsthm}
\usepackage{amsmath}
\usepackage{cleveref}
\usepackage{enumitem}
\usepackage{amsfonts}
\usepackage{amssymb}
\usepackage{xspace}
\usepackage{cite}

\newtheorem{theorem}{Theorem}
\newtheorem{algorithm}{Algorithm}
\newtheorem{example}{Example}
\newtheorem{corollary}{Corollary}
\newtheorem{definition}{Definition}
\newtheorem{lemma}{Lemma}

\definecolor{darkgreen}{rgb}{0.0,0.5,0.0}

\newcommand{\mds}{MDS\xspace}

\newcommand{\ds}{DS\xspace}
\newcommand{\smp}{SMP\xspace}
\newtheorem{examplesmpcont}{Example \smp continuation}

\title{Fully Lattice Linear Algorithms\thanks{Extended version of the paper with the same title to appear in The 24th International Symposium on Stabilization, Safety, and Security of Distributed Systems (SSS 2022).}}
\author{Arya Tanmay Gupta and Sandeep S Kulkarni\\\texttt{\{atgupta,sandeep\}@msu.edu}
}
\date{Computer Science and Engineering, Michigan State University
}

\begin{document}

\maketitle

\begin{abstract}

This paper focuses on analyzing and differentiating between lattice linear problems and algorithms. It introduces a new class of algorithms called \textit{(fully) lattice linear algorithms}. A property of these algorithms is that they induce a partial order among all states and form \textit{multiple lattices}. An initial state locks in one of these lattices. 
We present a lattice linear self-stabilizing algorithm for minimal dominating set.

\end{abstract}


\section{Introduction}

A multiprocessing system involves several processes running concurrently. These systems can provide a substantially larger computing power over a single processor. However, increased parallelism requires increased coordination, thereby increasing the execution time. 

The notion of detecting predicates to represent problems which induce partial order among the global states (lattice linear problems) was introduced in \cite{Garg2020}. If the states form a partial order, then the nodes can be allowed to read old data and execute asynchronously. In \cite{Gupta2021}, we introduced \textit{eventually lattice linear algorithms}, for problems where states do not naturally form a partial order (non-lattice linear problems), which induce a partial order among the feasible states.
In this paper, we differentiate between the partial orders in lattice linear problems and those induced by lattice linear algorithms in non-lattice linear problems.

The paper is organized as follows. In \Cref{section:preliminaries}, we elaborate the preliminaries and some background on lattice linearity. In \Cref{section:ll-algos}, we present a fully lattice linear algorithm for minimal dominating set. In \Cref{section:convergence-time}, we study the convergence time of algorithms traversing a lattice of states. We discuss related work in \Cref{section:literature} and conclude in \Cref{section:conclusion}. 

\section{Preliminaries and Background}\label{section:preliminaries}

In this paper, we are mainly interested in graph algorithms where the input is a graph $G$, $V(G)$ is the set of its nodes and $E(G)$ is the set of its edges. For a node $i\in V(G)$, $Adj_i$ is the set of nodes connected to $i$ by an edge, and $Adj^x_i$ are the set of nodes within $x$ hops from $i$, excluding $i$.

Each node in $V(G)$ stores a set of variables, which represent its \textit{local state}. A \textit{global state} is obtained by assigning each variable of each node a value from its respective domain. We use $S$ to denote the set of all possible global states. A global state $s\in S$ is represented as a vector, where $s[i]$ itself is a vector of the variables of node $i$.

Each node in $V(G)$ is associated with actions. Each action at node $i$ checks the values of nodes in $Adj_i^x\cup \{i\}$ (where the value of $x$ is problem dependent) and updates its own variables. A \textit{move} is an event in which some node $i$ updates its variables based on the variables of nodes in $Adj_i^x\cup \{i\}$.


$S$ is a \textit{lattice linear state space} if its states form a lattice. The nature of the partial order, present among the states in $S$ which makes it \textit{lattice linear}, is elaborated as follows. 
Local states are totally ordered and global states are partially ordered. 
We use '$<$' to represent both these orders. For a pair of global states $s$ and $s'$, $s<s'$ iff $(\forall i:(s[i]<s'[i]\lor s[i]=s'[i]))\land(\exists i:(s[i]<s'[i]))$.
We use the symbol `$>$' which is the opposite of `$<$', i.e. $s>s'$ iff $s' < s$.
In the lattice linear problems in \cite{Garg2020}, $s$ transitions to $s'$ where $s < s'$. 

\begin{definition}\label{definition:<-lattice}
    \textbf{{\boldmath$<$}-\textit{lattice}}. 
    The states in $S$ form $<$-lattice iff $\exists$ a relation $<$ s.t. 
    
    (1) if $(s, a, s') \in \delta$ then (a)  $\exists j:((s[j]<s^\prime[j])\land(\forall k.k\neq j: s[k]=s^\prime[k]))$, and (b) $s'$ is the closest state to $s$, i.e., $\not \exists b, s'' \  (s'' < s' \land (s, b, s'') \in \delta)$, and 
    
    (2) $\forall s$ such that $s$ is not a supremum, $\exists a, s': (s, a, s') \in \delta$.
    
    \noindent $S$ is a lattice linear state space iff its states $<$-lattice.
\end{definition}

\noindent\textbf{\textit{Remark}}: The definition of $<$-lattice, and the structure of lattices that we present in this paper, is strict, as it provides the underlying infrastructure for a traversing algorithm. Here, $s'$ is a parent of $s$ iff only one node takes an action, and that action should be of minimum possible magnitude. From the perspective of a distributed system, if $s$ is forbidden, then there will be some forbidden nodes in $s$. If only one forbidden node takes action, then the resulting state $s'$ is a parent of $s$, otherwise if more than one nodes take action, then the resulting state $s'$ is not a parent of $s$, but is reachable from $s$ through the lattice. In addition, a node can choose to take an action which is not of the minimum possible magnitude, this depends on the states that it reads from other nodes.

\section{Background}\label{section:background}

Certain problems can be represented by a predicate $\mathcal{P}$ such that for any node $i$, if $i$ is violating $\mathcal{P}$, then it must change its state, or else the system will not satisfy $\mathcal{P}$. If $i$ is violating $\mathcal{P}$ in some state $s$, then it is forbidden in $s$. Formally,

\begin{definition}\label{definition:forbidden-node}\cite{Garg2020}\small  $\textsc{Forbidden}(i,s,\mathcal{P})\equiv \lnot \mathcal{P}(s)\land(\forall s'>s:s'[i]=s[i]\implies\lnot \mathcal{P}(s'))$.
\end{definition}

The predicate $\mathcal{P}$ is \textit{lattice linear} with respect to the lattice induced in $S$ iff $s$ not being optimal implies that there is some forbidden node in $s$. Formally,

\begin{definition}\label{definition:llp}\cite{Garg2020}\textbf{Lattice Linear Predicate {\boldmath$\mathcal{P}$}.}
    $\forall s\in S:\textsc{Forbidden}(s,\mathcal{P})\implies \exists i:\textsc{Forbidden}(i,s,\mathcal{P})$.
\end{definition}

A problem $P$ is a \textit{lattice linear problem} iff it can be represented by a lattice linear predicate. Otherwise, $P$ is a \textit{non-lattice linear problem}. In this section, we discuss some background results on lattice linear problems.

We have from the above discussion that if the system is in a state $s$ and some node $i$ is forbidden in $s$, it is essential for $i$ to change its state in order for the system to reach an optimal state. If no node is forbidden in $s$, then $\mathcal{P}(s)$ holds true. The lowest state in the lattice where $\mathcal{P}(s)$ is true is optimal.

\noindent\textbf{\textit{Remark:}} It is not necessary that a $<$-lattice is of finite size, in which case, there will be no supremum. But in the case of lattice linear problems a $<$-lattice will always have an infimum. 

\begin{example}\label{example:mom-definition}\textit{SMP}.
    We describe a lattice linear problem, the \textit{stable (man-optimal) marriage problem} (\smp) from \cite{Garg2020}. In \smp, all men (respectively, women) rank women (respectively men) in terms of their preference (lower rank is preferred more).
    A global state is represented as a vector $s$ where the vector $s[i]$ represents the (local) state of man $i$. $s[i]$ contains a single value which represents the ID of the woman $w$ that man $i$ is proposing to.
    
    The requirements of SMP can be defined as $\mathcal{P}_{\smp}\equiv \forall m,m',m\neq m': s[m]\neq s[m']$. $\mathcal{P}_{\smp}$ is true iff no two men are proposing to the same woman. In addition, man $m$ is forbidden in state $s$ iff there exists $m'$ such that $m$ and $m'$ are proposing to the same woman $w$ and $w$ prefers $m'$ over $m$. Thus, \\
    {
    \small 
    $\textsc{Forbidden}_{\smp}(m,s,\mathcal{P}_{\smp})\equiv \exists m': s[m]=s[m']\land rank(s[m],m')<rank(s[m],m).$
    \centering
    }
    If $m$ is forbidden, he increments $s[m]$ by 1 until all his choices are exhausted.
    \qed 
    
\end{example}

A key observation from the stable (man-optimal) marriage problem (\smp) and other problems from \cite{Garg2020} is that 
they exhibit global infimum state $\ell$ and \textit{possibly} a global supremum state $u$ i.e. $\ell$  and $u$ are the states such that $\forall s\in S, \ell\leq s$ and $\forall s\in S, u \geq s$.
All the states in $S$ form a single lattice. 

\begin{examplesmpcont}
    As an illustration of \smp, consider the case where we have 3 men $m_1,m_2,m_3$ and 3 women $w_1,w_2,w_3$. 
    The lattice induced in this case is shown in \Cref{figure:smplattice}. In this figure every vector represents the global state $s$ such that $s[i]$ represents the index of woman that $m_i$ is proposing to $\forall i\in \{1,2,3\}$. 
    The acting algorithm begins in the state $(1, 1, 1)$ (i.e., each man starts with his first choice) and continues its execution in this lattice. 
    The algorithm terminates in the lowest state in the lattice where no node is forbidden.
    \qed
    
    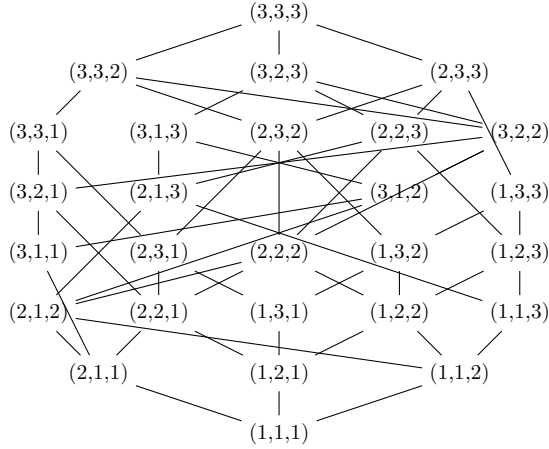
\begin{figure}[ht]
        \centering 
        \begin{tikzpicture}[scale=.8,every node/.style={scale=.8}]
            \node at (0,0) (a) {(1,1,1)};
            
            \node at (-3,1) (b) {(2,1,1)};
            \node at (0,1) (c) {(1,2,1)};
            \node at (3,1) (d) {(1,1,2)};
            
            \node at (4,2) (e) {(1,1,3)};
            \node at (2,2) (f) {(1,2,2)};
            \node at (4,3) (g) {(1,2,3)};
            
            \node at (0,2) (h) {(1,3,1)};
            \node at (2,3) (i) {(1,3,2)};
            \node at (4,4) (j) {(1,3,3)};
            
            \node at (-4,2) (k) {(2,1,2)};
            \node at (-2,4) (l) {(2,1,3)};
            
            \node at (-2,2) (m) {(2,2,1)};
            \node at (0,3) (n) {(2,2,2)};
            \node at (2,5) (o) {(2,2,3)};
            
            \node at (-2,3) (p) {(2,3,1)};
            \node at (0,5) (q) {(2,3,2)};
            \node at (3,6) (r) {(2,3,3)};
            
            \node at (-4,3) (s) {(3,1,1)};
            \node at (2,4) (t) {(3,1,2)};
            \node at (-2,5) (u) {(3,1,3)};
            
            \node at (-4,4) (v) {(3,2,1)};
            \node at (4,5) (w) {(3,2,2)};
            \node at (0,6) (x) {(3,2,3)};
            
            \node at (-4,5) (y) {(3,3,1)};
            \node at (-3,6) (z) {(3,3,2)};
            \node at (0,7) (parent) {(3,3,3)};
            
            \draw (a) -- (b);
            \draw (a) -- (c);
            \draw (a) -- (d);
            \draw (d) -- (e);
            \draw (d) -- (f);
            \draw (c) -- (f);
            \draw (e) -- (g);
            \draw (f) -- (g);
            \draw (c) -- (h);
            \draw (f) -- (i);
            \draw (h) -- (i);
            \draw (g) -- (j);
            \draw (i) -- (j);
            \draw (b) -- (k);
            \draw (d) -- (k);
            \draw (e) -- (l);
            \draw (k) -- (l);
            \draw (b) -- (m);
            \draw (c) -- (m);
            \draw (m) -- (n);
            \draw (n) -- (o);
            \draw (g) -- (o);
            \draw (l) -- (o);
            \draw (f) -- (n);
            \draw (k) -- (n);
            \draw (h) -- (p);
            \draw (m) -- (p);
            \draw (p) -- (q);
            \draw (i) -- (q);
            \draw (n) -- (q);
            \draw (o) -- (r);
            \draw (j) -- (r);
            \draw (q) -- (r);
            \draw (b) -- (s);
            \draw (k) -- (t);
            \draw (s) -- (t);
            \draw (t) -- (u);
            \draw (l) -- (u);
            \draw (s) -- (v);
            \draw (m) -- (v);
            \draw (n) -- (w);
            \draw (t) -- (w);
            \draw (v) -- (w);
            \draw (w) -- (x);
            \draw (u) -- (x);
            \draw (o) -- (x);
            \draw (p) -- (y);
            \draw (v) -- (y);
            \draw (y) -- (z);
            \draw (w) -- (z);
            \draw (q) -- (z);
            \draw (z) -- (parent);
            \draw (r) -- (parent);
            \draw (x) -- (parent);
        \end{tikzpicture}
        \caption{Lattice for \smp with 3 men and 3 women; $\ell=(1,1,1)$ and $u=(3,3,3)$.}
        \label{figure:smplattice}
    \end{figure}
\end{examplesmpcont}

\noindent\textbf{\textit{Remark}}: The lattices in \Cref{figure:smplattice} present the skeleton of the partial order among the states. Here, in a non-optimal state, if only one forbidden node moves, then the resulting state $s'$ is a parent of $s$. If more than one nodes move, then the resulting state $s'$ is not a parent of $s$, but is reachable from $s$ through the lattice.

In \smp and other problems in \cite{Garg2020}, the algorithm needs to be initialized to $\ell$ to reach an optimal solution.
If we start from a state $s$, $s \neq \ell$, then the algorithm can only traverse the lattice from $s$. Hence, upon termination, it is possible that the optimal solution is not reached. 
In other words, these algorithms cannot be self-stabilizing \cite{Dijkstra1974} unless $u$ is the optimal state. 

\begin{examplesmpcont}
    Consider that men and women are $M=(A,J,T)$ and $W=(K,Z,M)$ indexed in that sequence respectively. Let that proposal preferences of men are $A=(Z,K,M)$, $J=(Z,K,M)$ and $T=(K,M,Z)$, and women have ranked men as $Z=(A,J,T)$, $K=(J,T,A)$ and $M=(T,J,A)$. The optimal state (starting from (1,1,1)) is (1,2,2).
    Starting from (1,2,3), the algorithm terminates at (1,2,3) which is not optimal.
    Starting from (3,1,2), the algorithm terminates declaring that no solution is available.
    \qed 
\end{examplesmpcont}

\section{Lattice Linear \textit{Algorithms}: Minimal Dominating Set}\label{section:ll-algos}

In non-lattice linear problems such as minimal dominating set (\mds), the states do not form a partial order naturally, as for a given non-optimal state, it cannot be determined that which nodes are forbidden.

We introduce the class of fully lattice linear algorithms,
which partition the state space into subsets $S_1, S_2, \cdots,S_w (w\geq 1)$, where each subset forms a lattice. The initial state locks the system into one of these lattices and algorithm executes until an optimal state is reached. The optimal state is always the supremum of that lattice. In this section, we describe a fully lattice linear algorithm for \mds. 


\begin{definition}\label{definition:mds}\textbf{Minimal dominating set.}
    In the \mds problem, the task is to choose a minimal set of nodes $\mathcal{D}$ in a given graph $G$ such that for every node in $V(G)$, either it is present in $\mathcal{D}$, or at least one of its neighbours is in $\mathcal{D}$. Each node $i$ stores a variable $st.i$ with domain $\{IN, OUT\}$; $i\in\mathcal{D}$ iff $st.i=IN$.
\end{definition}

\noindent We describe the algorithm as \Cref{algorithm:ds-ll}.

\begin{algorithm}\label{algorithm:ds-ll}Algorithm for \mds.
    \begin{center}
    \small
    \begin{tabular}{|l|}
        \hline 
        \textsc{Removable-DS}$(i)$ $\equiv st.i=IN\land (\forall j\in Adj_i\cup\{i\}:((j\neq i\land st.j=IN)\lor$\\\quad\quad\quad\quad\quad\quad\quad\quad $(\exists k\in Adj_j, k\neq i: st.k=IN)))$.\\
        \textsc{Addable-DS}$(i)$ $\equiv st.i=OUT\land(\forall j\in Adj_i,st.j=OUT)$.\\
        \textsc{Unsatisfied-DS}$(i)$ $\equiv \textsc{Removable-DS}(i)\lor$
        $\textsc{Addable-DS}(i)$.\\
        $\textsc{Forbidden-DS}(i)\equiv\textsc{Unsatisfied-DS}(i)\land (\forall j\in Adj^2_i:$\\\quad\quad\quad\quad\quad\quad\quad\quad$\lnot\textsc{Unsatisfied-DS}(j) \lor id.i>id.j)$.\\
        Rules for node $i$:\\
        $\textsc{Forbidden-DS}(i) \longrightarrow st.i = \lnot st.i$.\\
        \hline 
    \end{tabular}
    \end{center}
\end{algorithm}

To demonstrate that \Cref{algorithm:ds-ll} is lattice linear, we define state value and rank, assumed as imaginary variables associated with the nodes, as follows:
{
\small 
\begin{equation*}
    \textsc{State-Value-DS}(i,s)=
    \begin{cases}
        1 & \text{if $\textsc{Unsatisfied-DS}(i)$ in state $s$} \\
        0 & \text{otherwise}
    \end{cases}
\end{equation*}

\begin{equation*}
    \textsc{Rank-DS}(s)=\sum_{i\in V(G)}\textsc{State-Value-DS}(i,s).
\end{equation*}
}

The lattice is formed with respect to the state value. Specifically, the state value of a node can change from $1$ to $0$ but not vice versa. 
Therefore $\textsc{Rank-DS}$ always decreases until it becomes zero at the supremum. 

\begin{lemma}\label{lemma:ds-no-step-back}
    Any node in an input graph does not revisit its older state while executing under \Cref{algorithm:ds-ll}.
\end{lemma}

\begin{proof}
    In \Cref{algorithm:ds-ll}, if a node $i$ is forbidden, then no node in $Adj^2_i$ moves.
    
    If $i$ is forbidden and addable at time $t$, then any other node in $Adj_i$ is out of the \ds. When $i$ moves in, then any other node in $Adj_i$ is no longer addable, so they do not move in after $t$. As a result $i$ does not have to move out after moving in. Similarly, a forbidden and removable $i$ does not move in after moving out.
    
    Let that $i$ is dominated and out, and $j\in Adj_i$ is removable forbidden. $j$ will move out only if $i$ is being covered by another node. Also, while $j$ turns out of the \ds, no other node in $Adj^2_j$, and consequently in $Adj_i$, changes its state. As a result $i$ does not have to turn itself in because of the action of $j$.
    
    From the above cases, we have that $i$ does not change its state to $st.i$ after changing its state from $st.i$ to $st'.i$ throughout the execution of \Cref{algorithm:ds-ll}.
\end{proof}

\begin{theorem}\label{theorem:ds-ll}
    \Cref{algorithm:ds-ll} is self-stabilizing and (fully) lattice linear.
\end{theorem}

\begin{proof}
    From \Cref{lemma:ds-no-step-back}, if $G$ is in state 
    $s$ and $\textsc{Rank-DS}(s)$ is non-zero, then at least one node is forbidden in $s$, so $\textsc{Rank-DS}$ decreases monotonously until it becomes zero. For any node $i$, we have that its state value decreases whenever $i$ is forbidden and never increases. Thus \Cref{algorithm:ds-ll} is self-stabilizing.
    
    We have a partial order among the states, where if the rank of a state $s$ is nonzero, then it transitions to a state $s'$ such that $s<s'$ where for some $i$ forbidden in $s$, $s[i]<s'[i]$. Here, $s<s'$ iff $\textsc{Rank-DS}(s)>\textsc{Rank-DS}(i)$.
    This 
    shows that \Cref{algorithm:ds-ll} is lattice linear.
\end{proof}

\begin{example}\label{example:fullylatticelinear}
    Let $G_4$ be a graph where $V(G_4)=\{v_1,v_2,v_3,v_4\}$ and $E(G_4)=\{\{v_1$, $v_2\}, \{v_3$, $v_4\}\}$. For $G_4$ the lattices induced under \Cref{algorithm:ds-ll} are shown in \Cref{figure:full-lattices-from-ds-example}; each vector represents a global state $(st.v_1, st.v_2, st.v_3,st.v_4)$. 
    \qed 
\end{example}
\begin{figure}[ht]
    \centering
    \subfigure[]{
        \begin{tikzpicture}[scale=.8,every node/.style={scale=.6}]
            \node at (0,0) (a1) {\begin{tabular}{c} (IN,OUT,\\ IN,OUT) \end{tabular}};
            \node at (-1.5,-1) (a2) {\begin{tabular}{c} (IN,OUT, \\ IN,IN) \end{tabular}};
            \node at (1.5,-1) (a3) {\begin{tabular}{c} (IN,IN, \\ IN,OUT) \end{tabular}};
            \node at (0,-2) (a4) {\begin{tabular}{c} (IN,IN, \\ IN,IN) \end{tabular}};
            \draw (a1) -- (a2);
            \draw (a1) -- (a3);
            \draw (a2) -- (a4);
            \draw (a3) -- (a4);
        \end{tikzpicture}
    }
    \subfigure[]{
        \begin{tikzpicture}[scale=.8,every node/.style={scale=.6}]
            \node at (0,0) (a1) {\begin{tabular}{c} (OUT,IN, \\ OUT,IN) \end{tabular}};
            \node at (-1.5,-1) (a2) {\begin{tabular}{c} (OUT,IN, \\ OUT,OUT) \end{tabular}};
            \node at (1.5,-1) (a3) {\begin{tabular}{c} (OUT,OUT, \\ OUT,IN) \end{tabular}};
            \node at (0,-2) (a4) {\begin{tabular}{c} (OUT,OUT, \\ OUT,OUT) \end{tabular}};
            \draw (a1) -- (a2);
            \draw (a1) -- (a3);
            \draw (a2) -- (a4);
            \draw (a3) -- (a4);
        \end{tikzpicture}
    }
    \subfigure[]{
        \begin{tikzpicture}[scale=.8,every node/.style={scale=.6}]
            \node at (0,0) (a1) {\begin{tabular}{c} (OUT,IN, \\ IN,OUT) \end{tabular}};
            \node at (-1.5,-1) (a2) {\begin{tabular}{c} (OUT,IN, \\ IN,IN) \end{tabular}};
            \node at (1.5,-1) (a3) {\begin{tabular}{c} (OUT,OUT, \\ IN,OUT) \end{tabular}};
            \node at (0,-2) (a4) {\begin{tabular}{c} (OUT,OUT, \\ IN,IN) \end{tabular}};
            \draw (a1) -- (a2);
            \draw (a1) -- (a3);
            \draw (a2) -- (a4);
            \draw (a3) -- (a4);
        \end{tikzpicture}
    }
    \subfigure[]{
        \begin{tikzpicture}[scale=.8,every node/.style={scale=.6}]
            \node at (0,0) (a1) {\begin{tabular}{c} (IN,OUT, \\ OUT,IN) \end{tabular}};
            \node at (-1.5,-1) (a2) {\begin{tabular}{c} (IN,IN, \\ OUT,IN) \end{tabular}};
            \node at (1.5,-1) (a3) {\begin{tabular}{c} (IN,OUT, \\ OUT,OUT) \end{tabular}};
            \node at (0,-2) (a4) {\begin{tabular}{c} (IN,IN, \\ OUT,OUT) \end{tabular}};
            \draw (a1) -- (a2);
            \draw (a1) -- (a3);
            \draw (a2) -- (a4);
            \draw (a3) -- (a4);
        \end{tikzpicture}
    }
    \caption{The lattices induced by \Cref{algorithm:ds-ll} on the graph $G_4$ described in \Cref{example:fullylatticelinear}.}
    \label{figure:full-lattices-from-ds-example}
\end{figure}
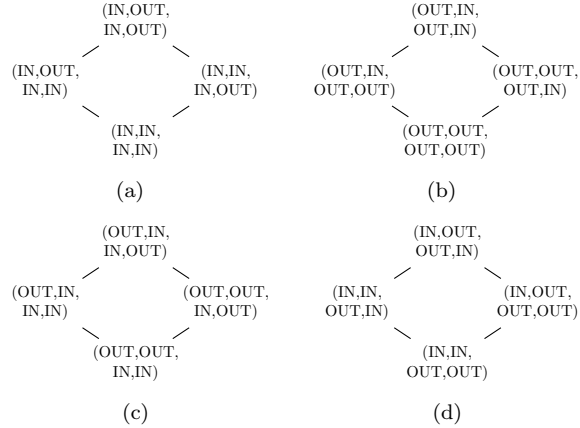

\noindent\textbf{\textit{Remark}}: The lattices in \Cref{figure:full-lattices-from-ds-example} present the skeleton of the partial order among the states. Here, in a non-optimal state, if only one forbidden node moves, then the resulting state $s'$ is a parent of $s$. If more than one nodes move, then the resulting state $s'$ is not a parent of $s$, but is reachable from $s$ through the lattice.

\section{Convergence Time in Traversing a Lattice of States}\label{section:convergence-time}

\begin{theorem}\label{theorem:general-convergence-time}
    Given an LLTS on $n$ processes, with the domain of size not more than $m$ for each process, the acting algorithm will converge in $n\times (m-1)$ moves.
\end{theorem}

\begin{proof}
    Assume for contradiction that the underlying algorithm converges in $x\geq n\times (m-1)+1$ moves. This implies, by pigeonhole principle, that at least one of the nodes $i$ is revising their states $st.i$ after changing to $st'.i$. If $st.i$ to $st'.i$ is a step ahead transition for $i$, then $st'.i$ to $st.i$ is a step back transition for $i$ and vice versa. For a system containing a lattice linear state space, we obtain a contradiction since step back actions are absent in such systems.
\end{proof}

\begin{corollary}\label{corollary:convergence-time-ll-multivariable}
    Consider the case where the nodes have multiple variables. 
    Furthermore, in each node, atmost $r$ of these variables,  $var_1.i,...,var_r.i$ (with domain sizes $m_1',...m_r'$ respectively) contribute independently to the construction of the lattice. 
    Then the LLTS will converge in $n\times \Big(\sum\limits_{j=1}^r (m_j'-1)\Big)$ moves.
    \qed 
\end{corollary}

\begin{corollary}
    (From \Cref{theorem:general-convergence-time} and \Cref{theorem:ds-ll}) \Cref{algorithm:ds-ll} converges in $n$ moves.
\end{corollary}

\noindent\textbf{\textit{Remark}}: In \cite{Gupta2021}, we presented an eventually lattice linear algorithm for the service demand based minimal dominating set problem, a more generalized version of the dominating set problem. That algorithm, along with the algorithms for minimal vertex cover, maximal independent set and graph colouring present in that paper, induces a partial order only among the feasible states. It reaches a feasible state aggressively where it locks in one of the lattices, and then it takes the system to an optimal state while traversing that lattice. Each of these two phases take $n$ moves. Thus that algorithm converges in $2n$ moves, which is shown in \cite{Gupta2021}.
Based on that algorithm, an algorithm for the standard dominating set problem (\Cref{definition:mds}) can be obtained, which will have the same properties. The lattices that it would induce will be similar to those drawn in \Cref{figure:full-lattices-from-ds-example}, except that non-feasible states will not be present in those lattices.

\section{Related Work}\label{section:literature}

\textbf{Lattice theory}: Lattice linear problems are studied in \cite{Garg2020}. In \cite{Gupta2021}, we have extended the theory presented in \cite{Garg2020} to develop eventually lattice linear self-stabilizing algorithms for some non-lattice linear problems. Such algorithms impose a lattice among the feasible states of the state space.

In this paper, we present a (fully) lattice linear algorithm for maximal dominating set, which imposes a partial order among all states and converges faster.

\noindent\textbf{Dominating set}: Self-stabilizing algorithms for the minimal dominating set problem are proposed in \cite{Turau2007,GODDARD2008,Chiu2014,Xu2003,Hedetniemi2003}. The best convergence time among these works is $4n$ moves. The algorithm presented in \cite{Gupta2021}, takes $2n$ moves to converge.

In this paper, the fully lattice linear algorithm that we present converges in $n$ moves and is fully tolerant to consistency violations. This is an improvement as compared to the results presented in the literature.

\section{Conclusion}\label{section:conclusion}

In this paper, we study the differences between the structure of partial order induced in lattice linear problems and non-lattice linear problems. We present a fully lattice linear self-stabilizing algorithm for the minimal dominating set. This is the first lattice linear algorithm for a non-lattice linear problem. This algorithm converges in $n$ moves.
We provide upper bounds to the convergence time for an algorithm traversing an arbitrary lattice linear state space.

It is still an open question whether fully lattice linear algorithms for minimal vertex cover and maximal independent set problems can be developed.

\bibliography{pvll.bib}
\bibliographystyle{acm}
\end{document}